\documentclass[journal,draftcls,onecolumn,12pt,twoside]{IEEEtran}

\usepackage{times}
 \usepackage{amsmath}
 \usepackage{amssymb}
 \usepackage{amsthm}
 \usepackage{color}
  \usepackage{algorithm}
 \usepackage[noend]{algorithmic}
      \usepackage{graphicx}
\usepackage{subfig}
\usepackage{multirow}
 \usepackage{url}
\usepackage{cite}
\usepackage{bm}

\newtheorem{theorem}{Theorem}

\newtheorem{lemma}[theorem]{Lemma}

\newtheorem{defn}{Definition}

\normalsize

\ifCLASSINFOpdf
\else
\fi

\begin{document}
%
\title{Cauchy MDS Array Codes With Efficient Decoding Method}

\author{Hanxu Hou and Yunghsiang S. Han,~\IEEEmembership{Fellow,~IEEE}
\thanks{H. Hou is with the School of Electrical Engineering \& Intelligentization, Dongguan University of Technology~(E-mail:houhanxu@163.com) and Y. S. Han is with the School of Electrical Engineering \& Intelligentization, Dongguan University of Technology and the Department of Electrical Engineering, National Taiwan University of Science and Technology~(E-mail: yunghsiangh@gmail.com).}
 }

\markboth{IEEE Transactions on Communications}%
{Submitted paper}

\maketitle
\vspace{-1.8cm}
\begin{abstract}
Array codes have been widely used in communication and storage systems. To reduce computational complexity, one important property of the array codes is that only XOR operation is used in the encoding and decoding process. In this work, we present a novel family of maximal-distance separable (MDS) array codes based on Cauchy matrix, which can correct up to any number of failures. We also propose an efficient decoding method for the new codes to recover the failures. We show that the encoding/decoding complexities of the proposed approach are lower than those of existing Cauchy MDS array codes, such as Rabin-Like codes and  CRS codes. Thus, the proposed MDS array codes are attractive for distributed storage systems.
\end{abstract}


\IEEEpeerreviewmaketitle

\section{Introduction}
\IEEEPARstart{A}rray codes are  error and burst correcting codes that have been widely employed in communication and storage systems~\cite{RAID89, RAID93} to  enhance data reliability. A common property of the array codes is that the encoding and decoding algorithms use only XOR  (exclusive OR) operations. A binary array code consists of an array of size $m \times n$, where each element in the array stores one bit. Among the $n$ columns (or data disks) in the array, the first $k$ columns are \emph{information columns} that store  \emph{information bits}, and the last $r$ columns are \emph{parity columns} that store \emph{parity bits}. Note that $n=r+k$. When a data disk fails, the corresponding column of the array code is considered as an \emph{erasure}. If the array code can tolerate arbitrary $r$ erasures, then it is named as a Maximum-Distance Separable (MDS) array code. In other words, in an MDS array code, the information bits can be recovered from any $k$ columns.

Besides the MDS property, the  performance of an MDS array code also depends on  encoding and decoding complexities. \emph{Encoding complexity} is defined as the number of XORs required to construct the parities and \emph{decoding complexity} is defined as the number of XORs required to recover the erased columns from any surviving $k$ ones.
The encoding and decoding procedures  of the array codes studied in most literature use simple XOR operations, that can be easily and most efficiently implemented. The MDS array codes proposed in this paper are also based on XOR operations.

\subsection{Related Work}
Row-diagonal parity (RDP) code proposed in~\cite{corbett2004row} and EVENODD code in~\cite{blaum1995evenodd} can tolerate two arbitrary disk erasures.
Due to increasing capacities of hard disks and requirement of  low bit error rates,
the protection offered by double parities will soon be inadequate. The issue of reliability is more pronounced in solid-state drives (SSD), which have significant wear-out rate when the frequencies of disk writes are high. Indeed, triple-parity RAID (Redundant Arrays of Inexpensive Disks) has already been advocated  in storage technology~\cite{BeyondRAID}.
Construction of array codes recovering multiple disk erasures is relatively rare, in compare to array codes
recovering double erasures. We name the existing MDS array codes in
\cite{corbett2004row,blaum1995evenodd,blaum1996mds,xiang2010optimal,blaum2001evenodd,huang2008star,wang2012triple,blaum2006family,feng2005new2}
as \emph{Vandermonde MDS array codes}, as their constructions are based on Vandermonde matrices.

Among the Vandermonde MDS array codes, BBV (Blaum, Bruck and Vardy) code~\cite{blaum1996mds,blaum2002evenodd}, which is an extension of the EVENODD code  for three or more parity columns, has the
best fault-tolerance. In \cite{blaum1996mds}, it is
proved that an extended BBV code is always an MDS code for three parity columns, but may not
be an MDS code for four or more party columns. A necessary and sufficient condition for the extended BBV code with four parity columns to be an MDS code
is given in \cite{blaum1996mds}, and some results for no more than eight parity columns are provided.

Another family of MDS array codes is called \emph{Cauchy MDS array codes}, which is constructed based on
Cauchy matrices. CRS codes in~\cite{Blomer1999An}, Rabin-like codes in~\cite{feng2005new2} and Circulant Cauchy codes
in~\cite{schindelhauer2013maximum} are examples of Cauchy MDS array codes.
Bl$\ddot{o}$mer {\em et al.} constructed CRS codes by employing a Cauchy matrix to perform encoding
(and upon failure, decoding) over a finite field instead of a binary field \cite{Blomer1999An}. In this approach,
the isomorphism and companion methods converting  a normal finite field operation to a binary field XOR
operation are necessary. The idea  is to replace an original symbol in the finite field with a matrix in
another finite field.
The authors in \cite{schindelhauer2013maximum} considered a special class of CRS codes,  called Circulant
Cauchy codes, that has lower encoding and decoding complexity than CRS codes. Based on the concept of permutation matrix, Feng \cite{feng2005new2}
gave a construction method to convert the Cauchy matrix to a sparse matrix. Compared with
the Vandermonde MDS array codes,  Cauchy
MDS array codes have better fault-tolerance at a cost of higher computational complexity. The purpose of this paper
is to give a new construction of Cauchy MDS array codes and to propose an efficient decoding method for the proposed codes.
\subsection{Cauchy Reed-Solomon Code}
Cauchy Reed-Solomon (CRS) code \cite{Blomer1999An} is one variant of RS codes that has better coding
complexity by employing Cauchy matrix. The key to constructing good CRS codes is the selection of Cauchy matrices.
Given $k$ data symbols in finite field $\mathbb{F}_{2^w}$, we can generate $r$ encoded symbols of a CRS code
with $k+r\leq 2^w$
in the following way. Let $X=\{x_1,\ldots,x_r\}$, $Y=\{y_1,\ldots,y_k\}$, where $X\cap Y=\emptyset$, such that each  $x_i$ and $y_i$ is a distinct
element in $\mathbb{F}_{2^w}$. The entry $(i,j)$ of the Cauchy matrix is calculated as $1/(x_i+y_j)$. Since
each element of $\mathbb{F}_{2^w}$ can be represented by a $w\times w$ binary matrix,
we can transform the $r\times k$ Cauchy matrix into an $rw\times kw$ \emph{binary distribution matrix}.
We divide each data symbol into $w$ trips. The $r$ encoded symbols are created by multiplying the $rw\times kw$
binary distribution matrix and $kw$ data strips. During the multiplication process, when there exists ``1'' in every row of
the binary distribution matrix, we do XOR operations on the corresponding data strips to obtain the strips of encoded
symbols.

\begin{figure}
\centering
\includegraphics[width=0.85\textwidth]{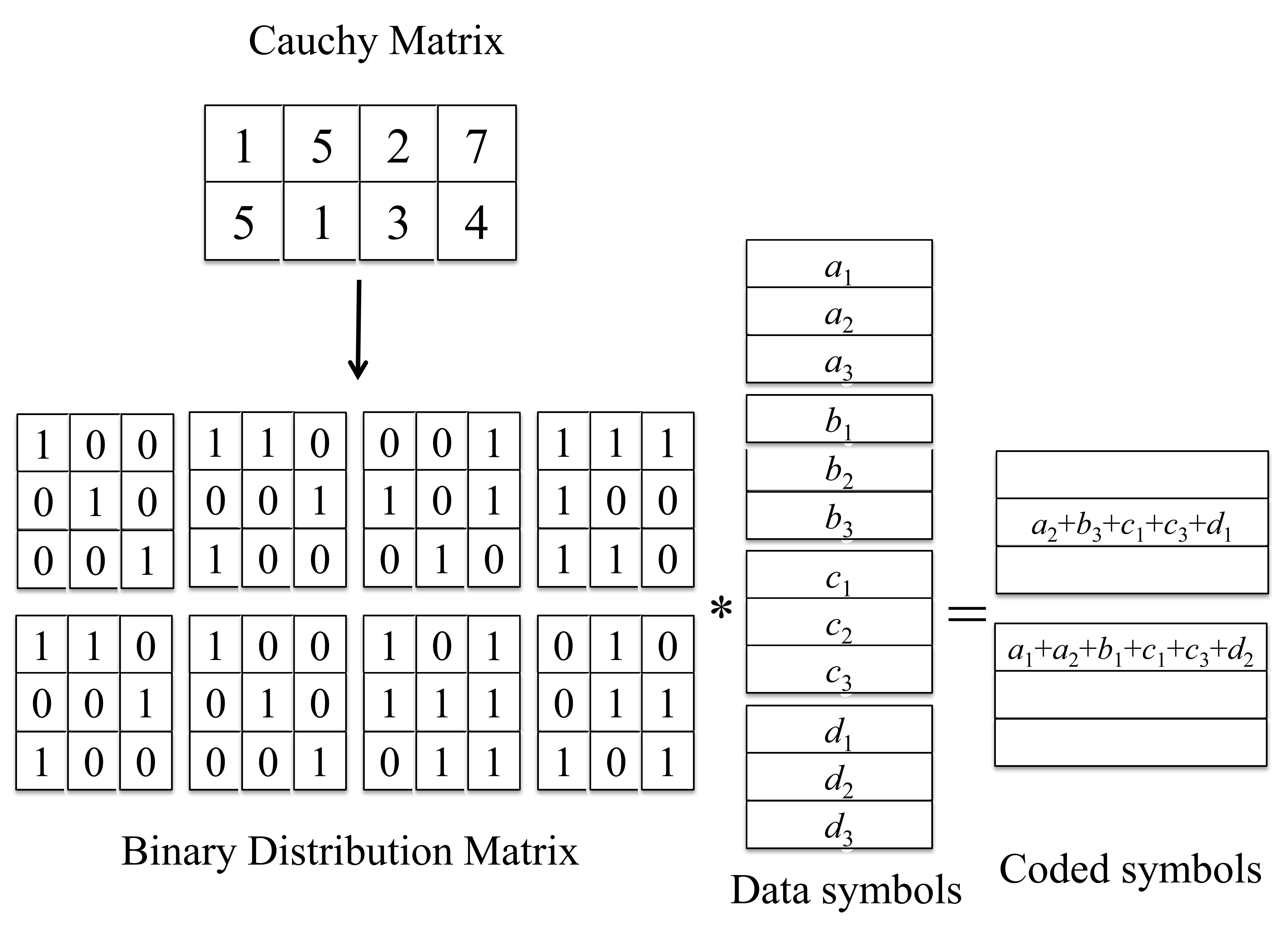}
\caption{The encoding process of encoded symbols of CRS codes with $k=4$,
$r=2$ and $w=3$ over $\mathbb{F}_8$.}
\label{example}
\end{figure}

With binary distribution matrix, one may create a strip of encoded symbol as the XOR of all data strips whose
corresponding columns of the binary distribution matrix have all ones. Note that, in this approach, the expensive matrix multiplication
is replaced by binary addition. Hence, this is a great improvement over standard RS codes. For more information
about the encoding and decoding process of CRS codes, please refer to \cite{Plank2006Optimizing}. Fig. \ref{example} displays the encoding process of encoded symbols for a CRS code with $k=4$, $r=2$, and
$w=3$ over $\mathbb{F}_8$. The first strip of the first encoded symbol and the second strip of the second encoded symbol may be calculated as
\begin{align*}
&a_2+b_3+c_1+c_3+d_1 \text{ and }\\
&a_1+a_2+b_1+c_1+c_3+d_2
\end{align*}
respectively. The two strips can be calculated by 9 XORs.

To improve the coding performance of a distributed storage system, one should reduce the number of XORs in the
coding processes. There are two approaches to achieve this goal.
\begin{enumerate}
\item \textbf{Choosing ``good'' Cauchy matrix}. Since the Cauchy matrix dictates the number of XORs~\cite{Plank2006Optimizing}, many researchers~\cite{Plank2006Optimizing,Blaum1999On,Plank2009A} had designed codes with low density Cauchy matrices. However, the only way to find lowest-density Cauchy matrices is to enumerate all the matrices and select the best one, where the number of matrices is exponential in $k$ and $r$. Therefore, this method is only feasible for  small $k$ and $r$. For example, when the parameters $k,r,w$ are small, the performance of CRS is optimized\cite{Plank2006Optimizing,plank2008jerasure}
by choosing the Cauchy matrix of which the corresponding binary distribution matrix has the lowest ``1''s.
\item \textbf{Encoding data using schedule.} The issue of exploiting common sums in the XOR equations is
addressed in \cite{Huang2005On,Plank2011XOR}. However, finding a good schedule with minimum XORs is still an open problem. Some heuristic schedules are proposed in \cite{Hafner2005Matrix,Yin2011Acoustic,Plank2012Heuristics,Zhang2015CaCo}. In the above example, the two strips containing  $c_1+c_3$ are treated as a subexpression. Therefore, if the bit $c_1+c_3$ is calculated before the calculation of two strips, then the two strips can be computed recursively by $x_1=c_1+c_3,x_2=a_2+b_3,x_3=x_2+x_1,x_4=x_3+d_1$ and $x_5=a_1+a_2,x_6=x_5+b_1,x_7=x_6+x_1,x_8=x_7+d_2$ with 8 XORs.
\end{enumerate}
\subsection{Contribution of This Paper}

In this paper, we present a new class of Cauchy MDS array codes, The proposed Cauchy MDS array codes $\mathcal{C}(k,r,p)$ are similar to CRS codes, except that $\mathcal{C}(k,r,p)$ are defined over a specific polynomial ring with a cyclic structure, rather than over a finite field. An efficient decoding algorithm is designed based on LU factorization of Cauchy matrix, which provides significant simplification of the decoding procedure for $\mathcal{C}(k,r,p)$. We demonstrate that the proposed $\mathcal{C}(k,r,p)$ has lowest encoding complexity and decoding complexity among the existing Cauchy array codes.

The rest of paper is organized as follows. In Section \ref{sec:array_code}, we give the construction of $\mathcal{C}(k,r,p)$ codes. After proving the MDS property of the proposed $\mathcal{C}(k,r,p)$ codes in Section~\ref{sec:MDS}, we give an efficient decoding method for any number of erasures in Section~\ref{sec:decode}. Section~\ref{sec:comparison} compares the computational complexity of encoding and decoding with the existing well-known MDS array codes, in terms of the number of XORs in computation. We conclude in
Section~\ref{sec:discussions}.

\section{A New Construction of Array Code}
\label{sec:array_code}
In this section, we will give a general construction of $\mathcal{C}(k,r,p)$ codes. Before that, we first introduce some facts on binary parity-check codes.

\subsection{Binary Parity-check Codes}
 A linear code $C$ over $\mathbb{F}_2$ is called a \emph{binary cyclic code} if, whenever $\textbf{\emph{c}}=(c_0,c_1,\ldots, c_{p-1})$ is in $C$, then $\textbf{\emph{c}}'=(c_{p-1},c_0,\ldots, c_{p-2})$ is also in $C$. The codeword $\textbf{\emph{c}}'$ is obtained by cyclically shifting the components of the codeword $\textbf{\emph{c}}$ one place to the right.
Let $p>2$ be a prime number and let $R_p$ be the ring
\begin{equation}
R_p \triangleq \mathbb{F}_2[x]/(1+x^p).
\label{eq:R}
\end{equation}
Every element of $R_p$ will be referred to as \emph{polynomial} in the sequel. The vector $(a_0,a_1,\ldots, a_{p-1}) \in \mathbb{F}_2^{p}$ is the codeword corresponding to the polynomial $\sum_{i=0}^{p-1} a_i x^i$. The indeterminate $x$ represents the \emph{cyclic-right-shift} operator on the codewords.  A subset of $R_p$ is a binary cyclic code of length $p$ if the subset is closed under addition and closed under multiplication by~$x$.

Consider the simple \emph{parity-check code}, $C_p$, which consists of polynomials in $R_p$ with even number of non-zero coefficients,
\begin{equation} C_p = \{ a(x) (1+x) \mod (1+x^p)|\, a(x) \in R_p \}.
\label{eq:C}
\end{equation}
The dimension of $C_p$ over $\mathbb{F}_2$ is $p-1$.  The {\em check polynomial} of $C_p$ is $h(x) = 1+x+\cdots + x^{p-1}$. That is, $\forall s(x)\in C_p$ and $c(x)\in R_p$, we have
\begin{equation}
s(x)(c(x)+h(x))=s(x)c(x) \mod (1+x^p),
\label{zero}
\end{equation}
since
\begin{align*}
s(x)h(x)&=(a(x)(1+x)\mod (1+x^p))h(x)\mod (1+x^p)\\
&=a(x)((1+x)h(x))\mod (1+x^p)\\
&= 0 \mod (1+x^p).
\end{align*}

Recall that, in a general ring $R$ with identity, there exists the identity $e$ such that $ue=eu=u$, $\forall u\in R$. The identity element of $C_p$ is
\[ e(x)\triangleq 1+h(x) = x + x^2 + \cdots + x^{p-1}. \]

We show that $C_p$ is isomorphic to $\mathbb{F}_2[x]/(h(x))$ in the next lemma.
\begin{lemma}
Let $p>2$ be a prime number, ring $C_p$ is isomorphic to $\mathbb{F}_2[x]/(h(x))$.
\label{lm:ismp}
\end{lemma}
\begin{proof}
We need to find an isomorphism between $C_p$ and $\mathbb{F}_2[x]/(h(x))$. Indeed, we can construct a function
$$\theta: C_p \rightarrow \mathbb{F}_2[x]/(h(x))$$
by defining
\[
 \theta(f(x)) \triangleq f(x) \bmod h(x).
\]
It is easy to check that $\theta$ is a ring homomorphism.
 Let us define $\phi(g(x))$  as
\[
\phi(g(x)) \triangleq (g(x) \cdot e(x)) \bmod (1+x^p).
\]
Next we prove that $\phi$ is an inverse function of $\theta$. For any polynomial $f(x)\in C_p$, if $\deg (f(x))<p-1$, then we have
\[
\phi(\theta(f(x)))=\phi(f(x))=f(x) \cdot e(x)=f(x) \bmod (1+x^p).
\]
Before we consider the case $\deg (f(x))=p-1$, we prove the following fact:
\begin{equation}
\label{eq:h2}
h(x)h(x)=h(x) \bmod (1+x^p).	
\end{equation}
Note that $h(x)$ can be reformulated as
$$h(x)=[1+x(1+x)+x^3(1+x)+\cdots+x^{p-2}(1+x)].$$
Hence
$$h(x)h(x)=h(x)+x(1+x)h(x)+x^3(1+x)h(x)+\cdots+x^{p-2}(1+x)h(x)=h(x) \bmod (1+x^p).$$
If $\deg (f(x))=p-1$, we have
\begin{eqnarray*}
\phi(\theta(f(x)))&=&\phi(f(x)+h(x))\\
&=&(f(x)+h(x)) \cdot (1+h(x))\\
&=&f(x)+f(x)h(x)+h(x)+h(x)h(x)\\
&=&f(x)+h(x)+h(x)h(x)\mbox{ (by the fact }f(x)h(x)=0 \bmod (1+x^p))\\
&=&f(x)+h(x)+h(x)\mbox{ (by \eqref{eq:h2})}\\
&=&f(x) \bmod (1+x^p).
\end{eqnarray*}
The composition $\phi \circ \theta$ is thus the identity mapping of~$C_p$ and  the mapping $\theta$ is a bijection. Therefore, $C_p$ is isomorphic to $\mathbb{F}_2[x]/(h(x))$.
\end{proof}
Note that $C_p$ is isomorphic to a finite field $\mathbb{F}_{2^{p-1}}$ if and only if 2 is a primitive element in $\mathbb{F}_p$~\cite{Fenn1997Bit}.
For example, when $p=5$, $C_5$ is isomorphic to a finite field $\mathbb{F}_{2^{4}}$ and the element $1+x^4$ in $C_p$ is mapped to
\[1+x^4 \bmod h(x)= x+x^2+x^3.
\]
If we apply the function $\phi$ to $x+x^2+x^3$, we can recover
\begin{align*}
\phi (x+x^2+x^3) &= (x+x^2+x^3)(x+x^2+x^3+x^4) \\
&= 1+x^4 \bmod (1+x^5).
\end{align*}

A polynomial $f(x)\in C_p$ is called \emph{invertible} if we can find a polynomial $\bar{f}(x)\in C_p$ such that $f(x)\bar{f}(x)$ is equal to the identity polynomial $e(x)$. The polynomial $\bar{f}(x)$ is called inverse of $f(x)$. It can be shown that the inverse  is unique in $C_p$.

%

The next lemma demonstrates that the polynomial $x^t+x^{t+b}$ is invertible.

\begin{lemma}
Let $p>2$ be a prime number, there exists a polynomial $a(x)\in C_p$ such that
\begin{equation}
a(x)(x^t+x^{t+b})= e(x) \mod (1+x^p),
\label{inverse}
\end{equation}
where $t, b\ge 0$ and $1\le t+b<p$.
\label{lm:inverse}
\end{lemma}
\begin{proof}
We can check that, in ring $R_p$,
\begin{align*}
&[x^t(1+x^b)][x^{p-t}(1+x^{2b}+x^{4b}+\cdots + x^{(p-1)b})] \bmod (1+x^p)\\
&=x^p[(1+x^b)(1+x^{2b}+x^{4b}+\cdots + x^{(p-1)b})] \bmod (1+x^p)\\
&=1+x^b+x^{2b}+x^{3b}+x^{4b}+\cdots +x^{(p-2)b}+x^{(p-1)b}+1 \bmod (1+x^p) \\
&=1+x+x^2+x^3+x^4+\cdots+x^{p-2}+x^{p-1}+1 \bmod (1+x^p)\\
&= e(x) \bmod (1+x^p).
\end{align*}
The second last equality follows from the fact that $\ell b\neq 0 \mod p$ for $(b,p)=1$ and $1\leq \ell \leq p-1$.
If $x^{p-t}(1+x^{2b}+x^{4b}+\cdots + x^{(p-1)b})\in C_p$, then the inverse of $x^t+x^{t+b}$ is $x^{p-t}(1+x^{2b}+x^{4b}+\cdots + x^{(p-1)b})$; Otherwise, the inverse of $x^t+x^{t+b}$ is $x^{p-t}(1+x^{2b}+x^{4b}+\cdots + x^{(p-1)b})+h(x)$ due to \eqref{zero}.
\end{proof}




In the following of the paper, we represent the inverse of $x^t+x^{t+b}$ as $1/(x^t+x^{t+b})$. The following we present some properties of the inverses.

\begin{lemma}
Let $a,b,c,d$ be integers between 0 and $p-1$ such that $a\neq b$ and $c\neq d$. For two polynomials $s_1(x),s_2(x)\in R_p$, the following equations hold:
\begin{equation}
\frac{1}{x^a+x^b}\cdot \frac{1}{x^c+x^d}=\frac{1}{(x^a+x^b)(x^c+x^d)},
\label{eq:1}
\end{equation}
\begin{equation}
\frac{s_1(x)}{x^a+x^b}+ \frac{s_2(x)}{x^a+x^b}=\frac{s_1(x)+s_2(x)}{x^a+x^b},
\label{eq:2}
\end{equation}
and
\begin{equation}
\frac{1}{x^a+x^b}+ \frac{1}{x^c+x^d}=\frac{x^a+x^b+x^c+x^d}{(x^a+x^b)(x^c+x^d)}.
\label{eq:3}
\end{equation}
\end{lemma}
\begin{proof}
Let $p(x)$ and $q(x)$ be inverses of $x^a+x^b$ and $x^c+x^d$, respectively. That is,  $(x^a+x^b)p(x)=e(x)\bmod {(1+x^p)}$ and $(x^c+x^d)q(x)=e(x)\bmod {(1+x^p)}$. Thus we have
\[
(x^a+x^b)(x^c+x^d)p(x)q(x)=e(x)e(x)=e(x)\bmod {(1+x^p)}.
\]
Note that, $(x^a+x^b)(x^c+x^d)\in C_p$ and $p(x)q(x)\in C_p$. By definition, we have
\[
\frac{1}{(x^a+x^b)(x^c+x^d)}=p(x)q(x).
\]
Therefore, \eqref{eq:1} holds.

\eqref{eq:2} follows from
\[
p(x)s_1(x)+p(x)s_2(x)=p(x)(s_1(x)+s_2(x)).
\]

The right side of equation in \eqref{eq:3} is
\begin{align*}
\frac{(x^a+x^b)+(x^c+x^d)}{(x^a+x^b)(x^c+x^d)}&=\frac{(x^a+x^b)+(x^c+x^d)}{(x^a+x^b)}\cdot \frac{1}{(x^c+x^d)} \text{ (by \eqref{eq:1})} \\
&=(e(x)+\frac{x^c+x^d}{x^a+x^b})\cdot \frac{1}{(x^c+x^d)} \text{ (by \eqref{eq:2})} \\
&=\frac{1}{x^a+x^b}+ \frac{1}{x^c+x^d}.
\end{align*}
\end{proof}
For a square matrix in $C_p$, we define the inverse matrix as follows.
\begin{defn}
An $\ell\times \ell$ matrix $\mathcal{M}_{\ell\times\ell}$ is called \emph{invertible} if we can find an $\ell\times\ell$ matrix $\mathcal{M}^{-1}_{\ell\times\ell}$ such that $\mathcal{M}_{\ell\times\ell}\cdot \mathcal{M}^{-1}_{\ell\times\ell}=\mathcal{I}_{\ell}$, where $\mathcal{I}_{\ell}$ is the $\ell\times \ell$ identity matrix
\begin{align}
\label{eq:identity}
\mathcal{I}_{\ell}\triangleq\begin{bmatrix}e(x)&0&\cdots & 0\\
0&e(x)&\cdots & 0\\
\vdots & \vdots & \ddots & \vdots\\
0&0&\cdots &e(x)\\
\end{bmatrix}.
\end{align}
The matrix $\mathcal{M}^{-1}_{\ell\times\ell}$ is called the inverse matrix of $\mathcal{M}_{\ell\times\ell}$.
\end{defn}

\subsection{Construction of Cauchy Array Codes}
\label{subsec:construction}
In this subsection, we define a $(p-1)\times(k+r)$ array code, called $\mathcal{C}(k,r,p)$, where $p>2$ is a prime number and $p\geq k+r$. We index the columns by $\{0,1,\ldots, k+r-1\}$, and the rows by $\{0,1,\ldots, p-2\}$.
The columns are identified with the disks. Columns 0 to $k-1$ are called the \emph{information columns}, which store the information bits. Columns $k$ to $k+r-1$ are called the \emph{parity columns}, which store the redundant bits.

For $i=0,1,\ldots,p-2$ and $j=0,1,\ldots,k-1$, let $i$-th information bit in  $j$-th information column be denoted by $s_{i,j}$. For each $p-1$ information bits $s_{0,j},s_{1,j},\ldots,s_{p-2,j}$ stored in $j$-th information column,  one extra \emph{parity-check bit} $s_{p-1,j}$ is computed as
\begin{equation}\label{eq:parity-check-sj}
s_{p-1,j}\triangleq s_{0,j}+s_{1,j}+\cdots+s_{p-2,j}.
\end{equation}
Define \emph{data polynomial} for $j$-th information column as
\begin{equation}
s_j(x)\triangleq s_{0,j}+s_{1,j}x+\cdots+s_{p-2,j}x^{p-2}+s_{p-1,j}x^{p-1}.
\label{eq:sj}
\end{equation}
Note that the extra parity-check bit is not stored, and can be computed  when necessary. It is easy to see that each data polynomial is an element in $C_p$.


Next we present the method to compute the encoded symbols in parity columns. For $i=0,1,\ldots,p-2$ and $j=0,1,\ldots,r-1$, let $i$-th redundant bit stored in $j$-th parity column be denoted by $c_{i,j}$. Define \emph{coded polynomial} for $j$-th parity column as
\begin{equation}
c_{j}(x)\triangleq c_{0,j}+c_{1,j}x+\cdots+c_{p-2,j}x^{p-2}+c_{p-1,j}x^{p-1}.
\label{eq:cj}
\end{equation}
It will be clear later that
\begin{equation}
c_{p-1,j}=c_{0,j}+c_{1,j}+\cdots+c_{p-2,j}.\label{encoded-sum}
\end{equation}

The coded polynomial can be generated by
\begin{equation}
\begin{bmatrix}
c_0(x) & c_1(x) & \cdots & c_{r-1}(x)
\end{bmatrix}\triangleq
\begin{bmatrix}
s_0(x) & s_1(x) & \cdots & s_{k-1}(x)
\end{bmatrix}\cdot \mathbf{C}_{k\times r},
 \label{eq:en}
  \end{equation}
where
  \begin{equation}
  \mathbf{C}_{k\times r} \triangleq
 \begin{bmatrix}
\frac{1}{1+x^r} & \frac{1}{x+x^{r}} &  \cdots & \frac{1}{x^{r-1}+x^{r}} \\
\frac{1}{1+x^{r+1}} & \frac{1}{x+x^{r+1}} &  \cdots & \frac{1}{x^{r-1}+x^{r+1}} \\
\vdots & \vdots & \ddots& \vdots \\
\frac{1}{1+x^{k+r-1}} & \frac{1}{x+x^{k+r-1}} &  \cdots & \frac{1}{x^{r-1}+x^{k+r-1}}
\end{bmatrix}
 \label{eq:C}
  \end{equation}
is a $k\times r$ rectangular Cauchy matrix over $C_p$. Note that each entry of the matrix in \eqref{eq:C} is the inverse of $x^t+x^{t+b}$ and all arithmetic operations in \eqref{eq:en} are performed in ring $C_p$. The coded polynomials $c_j(x)$ in~\eqref{eq:en}, for $0\le j\le r-1$, are in $C_p$. Hence, the $k\times (k+r)$ generator matrix $\mathbf{G}_{k\times (k+r)}$ of the codewords
\[s_0(x),s_1(x),\ldots,s_{k-1}(x),c_0(x),c_1(x),\cdots,c_{r-1}(x)\]
is given by
\[\mathbf{G}_{k\times (k+r)}=\begin{bmatrix}\mathcal{I}_{k}& |& \mathbf{C}_{k\times r}\\ \end{bmatrix},\]
where $\mathcal{I}_{k}$ is the $k\times k$  matrix given in~\eqref{eq:identity}. Note that all entries in $\mathbf{G}_{k\times (k+r)}$ are in $C_p$.



The above encoding procedure can be summarized as three steps: (i) given $k(p-1)$ information bits, append $k$ extra
parity-check bits as given in \eqref{eq:parity-check-sj} and obtain the $k$ polynomials $$s_0(x), s_1(x),\cdots,s_{k-1}(x);$$ (ii)
generate $r$ coded polynomials as given in~\eqref{eq:en}; (iii) ignore the terms with degree $p-1$ of the coded
polynomials and store the coefficients of the terms in the coded polynomials of degrees from 0 to $p-2$.

Before we present a fast decoding algorithm to generate codewords, we first present the MDS property of the proposed array codes.

\section{The MDS Property}
\label{sec:MDS}

A $(p-1)\times n$ array code that encodes $k(p-1)$ information bits is said to be an {\em MDS array code} if the $k(p-1)$ information bits can be recovered by downloading  any $k$ columns.\footnote{In total, one needs to download $k\times (p-1)$ bits.} In this section, we are going to prove that the array code constructed in the last section satisfies the MDS property for $k+r\leq p$.

Next lemma shows a sufficient MDS property condition of the array code $\mathcal{C}(k,r,p)$.

\begin{lemma}
If any $k\times k$ sub-matrix of the generator matrix $\mathbf{G}_{k\times (k+r)}$, after reduction modulo $h(x)$, is a nonsingular matrix over $\mathbb{F}_2[x]/(h(x))$, then the array code $\mathcal{C}(k,r,p)$ satisfies the MDS property.
\label{thm:MDS1}
\end{lemma}
\begin{proof}
Recall that, according to Lemma \ref{lm:ismp}, $C_p$ is isomorphic to ring $\mathbb{F}_2[x]/(h(x))$.  Let $\mathbf{A}$ be a $k\times k$ sub-matrix of the generator matrix $\mathbf{G}_{k\times (k+r)}$, and $\bar{\mathbf{A}}$ be the matrix obtained by reducing each entry of $\mathbf{A}$ mod $h(x)$. Matrix $\bar{\mathbf{A}}$ can be regarded as a matrix over  $\mathbb{F}_2[x]/(h(x))$. Since  $\bar{\mathbf{A}}$ is nonsingular over $\mathbb{F}_2[x]/(h(x))$, we can find the inverse of $\bar{\mathbf{A}}$. Let $\bar{\mathbf{A}}^{-1}$ be the inverse of $\mathbf{A}$ as a matrix over  $\mathbb{F}_2[x]/(h(x))$, then we can compute the inverse matrix of $\mathbf{A}$ over $C_p$ by applying the inverse function $\phi$ for each entry of $\bar{\mathbf{A}}^{-1}$. Therefore, the array code $\mathcal{C}(k,r,p)$ satisfies the MDS property.
\end{proof}
With Lemma \ref{thm:MDS1}, we have that if the determinant of the $k\times k$ sub-matrix $\mathbf{A}$ of $\mathbf{G}_{k\times (k+r)}$ is invertible, then we can find a matrix $\mathbf{A}^{-1}$ such that $\mathbf{A}^{-1}\mathbf{A}=\mathcal{I}_{k}$, i.e., $\mathbf{A}^{-1}$ is the inverse matrix of $\mathbf{A}$.
We need the following result about the Cauchy determinant in ring $C_p$ before giving a characterization of the MDS property in terms of determinants.

\begin{lemma}
Let $x^{a_1},x^{a_2},\ldots,x^{a_\ell},x^{b_1},x^{b_2},\ldots,x^{b_\ell}$ be $2\ell$ distinct monomials, where $0\leq a_i,b_i<p$ for $i=1,2,\ldots,\ell$ and $p$ is a prime number.
The determinant of the Cauchy matrix
\begin{equation}
\mathbf{C}(x^{a_{1:\ell}},x^{b_{1:\ell}}) \triangleq
\begin{bmatrix}
\frac{1}{x^{a_1}+x^{b_1}} & \frac{1}{x^{a_1}+x^{b_2}} &  \cdots & \frac{1}{x^{a_1}+x^{b_\ell}} \\
\frac{1}{x^{a_2}+x^{b_1}} & \frac{1}{x^{a_2}+x^{b_2}} &  \cdots & \frac{1}{x^{a_2}+x^{b_\ell}} \\
\vdots & \vdots & \ddots& \vdots \\
\frac{1}{x^{a_\ell}+x^{b_1}} & \frac{1}{x^{a_\ell}+x^{b_2}} &  \cdots & \frac{1}{x^{a_\ell}+x^{b_\ell}}
\end{bmatrix}
\label{eq:cauchy}
\end{equation}
over $C_p$ is
\begin{equation}
D_{\ell}(x)=\frac{\prod_{\ell\geq j>i\geq 1}(x^{a_j}+x^{a_i})(x^{b_i}+x^{b_j})}{\prod_{\ell\geq j,i\geq 1}(x^{a_i}+x^{b_j})}.
\label{eq:detC}
\end{equation}
\label{lm:detC}
\end{lemma}
\begin{proof}
Recall that the polynomial $x^i+x^{j}$ is invertible in $C_p$ for $0\leq i<j<p$, and $\frac{1}{x^i+x^{j}}$ is the inverse of $x^i+x^{j}$. For the determinant of the Cauchy matrix in \eqref{eq:cauchy}, adding column 1 to each of columns 2 to $\ell$, we have the entry in the $i$ row and the $j$ column as
\begin{align*}
\frac{1}{x^{a_i}+x^{b_j}}+\frac{1}{x^{a_i}+x^{b_1}}=&\frac{(x^{a_i}+x^{b_1})+(x^{a_i}+x^{b_j})}{(x^{a_i}+x^{b_j})(x^{a_i}+x^{b_1})} \text{ (by \eqref{eq:3})}\\
=&\frac{(x^{b_j}+x^{b_1})}{(x^{a_i}+x^{b_1})}\cdot\frac{1}{(x^{a_i}+x^{b_j})} \text{ (by \eqref{eq:1})},
\end{align*}
where $1\leq i\leq \ell$ and $2\leq j\leq \ell$. There is no effect on the value of the determinant from multiple of row added to row of determinant. Thus, the determinant is
\[
D_{\ell}(x)=\begin{vmatrix}
\frac{1}{x^{a_1}+x^{b_1}} & \frac{(x^{b_2}+x^{b_1})}{(x^{a_1}+x^{b_1})}\cdot\frac{1}{(x^{a_1}+x^{b_2})} &  \cdots & \frac{(x^{b_\ell}+x^{b_1})}{(x^{a_1}+x^{b_1})}\cdot\frac{1}{(x^{a_1}+x^{b_\ell})}\\
\frac{1}{x^{a_2}+x^{b_1}} & \frac{(x^{b_2}+x^{b_1})}{(x^{a_2}+x^{b_1})}\cdot\frac{1}{(x^{a_2}+x^{b_2})} &  \cdots & \frac{(x^{b_\ell}+x^{b_1})}{(x^{a_2}+x^{b_1})}\cdot\frac{1}{(x^{a_2}+x^{b_\ell})} \\
\vdots & \vdots & \ddots& \vdots \\
\frac{1}{x^{a_\ell}+x^{b_1}} & \frac{(x^{b_2}+x^{b_1})}{(x^{a_\ell}+x^{b_1})}\cdot\frac{1}{(x^{a_\ell}+x^{b_2})} &  \cdots & \frac{(x^{b_\ell}+x^{b_1})}{(x^{a_\ell}+x^{b_1})}\cdot\frac{1}{(x^{a_\ell}+x^{b_\ell})}
\end{vmatrix}.
\]
Extracting the factor $\frac{1}{x^{a_i}+x^{b_1}}$ from the $i$ row for $i=1,2,\ldots,\ell$, and  the factor $x^{b_j}+x^{b_1}$ from the $j$ column for $j=2,3,\ldots,\ell$, we have
\[
D_{\ell}(x)=\left(\prod_{i=1}^{\ell}\frac{1}{x^{a_i}+x^{b_1}}\right)
\left(\prod_{j=2}^{\ell}(x^{b_j}+x^{b_1})\right)\begin{vmatrix}
1 & \frac{1}{(x^{a_1}+x^{b_2})} &  \cdots & \frac{1}{(x^{a_1}+x^{b_\ell})}\\
1 & \frac{1}{(x^{a_2}+x^{b_2})} &  \cdots & \frac{1}{(x^{a_2}+x^{b_\ell})} \\
\vdots & \vdots & \ddots& \vdots \\
1 & \frac{1}{(x^{a_\ell}+x^{b_2})} &  \cdots & \frac{1}{(x^{a_\ell}+x^{b_\ell})}
\end{vmatrix}.
\]
For $i=2,3,\ldots,\ell$, adding the first row to rows 2 to $\ell$, we have
\[
D_{\ell}(x)=\begin{vmatrix}
1 & \frac{1}{(x^{a_1}+x^{b_2})} &  \cdots & \frac{1}{(x^{a_1}+x^{b_\ell})}\\
0& \frac{(x^{a_1}+x^{a_2})}{(x^{a_1}+x^{b_2})}\cdot\frac{1}{(x^{a_2}+x^{b_2})} &  \cdots & \frac{(x^{a_1}+x^{a_2})}{(x^{a_1}+x^{b_\ell})}\cdot\frac{1}{(x^{a_2}+x^{b_\ell})} \\
\vdots & \vdots & \ddots& \vdots \\
0 & \frac{(x^{a_1}+x^{a_\ell})}{(x^{a_1}+x^{b_2})}\cdot\frac{1}{(x^{a_\ell}+x^{b_2})} &  \cdots & \frac{(x^{a_1}+x^{a_\ell})}{(x^{a_1}+x^{b_\ell})}\cdot\frac{1}{(x^{a_\ell}+x^{b_\ell})}
\end{vmatrix}.
\]
Again, extracting the factor $x^{a_1}+x^{a_i}$ from the $i$ row for $i=2,3,\ldots,\ell$, and  the factor $\frac{1}{x^{a_1}+x^{b_j}}$ from the $j$ column for $j=2,3,\ldots,\ell$, we have
\begin{align*}
D_{\ell}(x)=&\left(\prod_{i=1}^{\ell}\frac{1}{x^{a_i}+x^{b_1}}\right)\left(\prod_{j=2}^{\ell}\frac{1}{x^{a_1}+x^{b_j}}\right)
\left(\prod_{j=2}^{\ell}(x^{b_j}+x^{b_1})\right)\left(\prod_{i=2}^{\ell}(x^{a_1}+x^{a_i})\right)\\
&\cdot\begin{vmatrix}
1 & 1 &  \cdots & 1\\
0 & \frac{1}{(x^{a_2}+x^{b_2})} &  \cdots & \frac{1}{(x^{a_2}+x^{b_\ell})} \\
\vdots & \vdots & \ddots& \vdots \\
0 & \frac{1}{(x^{a_\ell}+x^{b_2})} &  \cdots & \frac{1}{(x^{a_\ell}+x^{b_\ell})}
\end{vmatrix}\\
=&\frac{\prod_{i=2}^{\ell}(x^{a_i}+x^{a_1})(x^{b_i}+x^{b_1})}{\prod_{1\leq i,j\leq \ell}(x^{a_i}+x^{b_1})(x^{a_1}+x^{b_j})}
\cdot\begin{vmatrix}
\frac{1}{(x^{a_2}+x^{b_2})} & \frac{1}{(x^{a_2}+x^{b_3})} &  \cdots & \frac{1}{(x^{a_2}+x^{b_\ell})}\\
\frac{1}{(x^{a_3}+x^{b_2})} & \frac{1}{(x^{a_3}+x^{b_3})} &  \cdots & \frac{1}{(x^{a_3}+x^{b_\ell})}\\
\vdots & \vdots & \ddots& \vdots \\
\frac{1}{(x^{a_\ell}+x^{b_2})} & \frac{1}{(x^{a_\ell}+x^{b_3})} &  \cdots & \frac{1}{(x^{a_\ell}+x^{b_\ell})}
\end{vmatrix}.
\end{align*}
Repeating the above process for the remaining $(\ell-1)\times(\ell-1)$ Cauchy determinant, we can obtain the determinant given in \eqref{eq:detC}.
\end{proof}
The next lemma gives a characterization of the MDS property in terms of determinants.
\begin{lemma}
Let $p$ be a prime number with $k+r \leq p$.
Then the determinant of any $k\times k$ sub-matrix of the generator matrix $\mathbf{G}_{k\times (k+r)}$, after reduction modulo $h(x)$, is invertible over $\mathbb{F}_2[x]/(h(x))$.
\label{invertible}
\end{lemma}
\begin{proof}
Note that the determinant of any square matrix of $\mathbf{G}_{k\times (k+r)}$ after reduction modulo $h(x)$ can be computed by first reducing each entry of the square matrix by $h(x)$, and then computing the determinant by reducing $h(x)$. It is sufficient to show that the determinant of any $\ell\times \ell$ sub-matrix of the matrix $\mathbf{C}_{k\times r}$, after reduction modulo $h(x)$, is invertible over $\mathbb{F}_2[x]/(h(x))$, for $1\leq \ell \leq \min\{k,r\}$.
Considering the matrix $\mathbf{C}_{k\times r}$ given in~\eqref{eq:C}, for any $\ell$ distinct rows indexed by $a_1,a_2,\cdots,a_{\ell}$ between 0 to $r-1$ and any $\ell$ distinct columns indexed by $b_1,b_2,\cdots,b_{\ell}$ between $r$ to $k+r-1$, the corresponding $\ell\times \ell$ sub-matrix is with the form of the Cauchy matrix given in \eqref{eq:cauchy}. Hence, the determinant is the polynomial given in \eqref{eq:detC}.
As the polynomial $x^i+x^{j}$ is invertible in $C_p$, where $0\leq i<j<p$,  \eqref{eq:detC} is invertible in $C_p$. By the definition of invertible, there exist a polynomial $a(x)\in C_p$ such that
\[D_{\ell}(x)a(x)=e(x)\mod (1+x^p)\]
holds. Therefore, we have
\[D_{\ell}(x)a(x)+(1+x^p)b(x)=1+h(x) \]
for some polynomial $b(x)\in C_p$, and
$$D_{\ell}(x)a(x)+h(x)((1+x)b(x)+1)=1.$$
Hence, $D_{\ell}(x)a(x)=1\bmod h(x)$ and
this proves that the polynomial $(D_{\ell}(x) \bmod h(x))$ is invertible in $\mathbb{F}_2[x]/(h(x))$.

\end{proof}

By applying Lemma~\ref{thm:MDS1} and Lemma~\ref{invertible}, we have the following theorem.
\begin{theorem}
Let $p>2$ be a prime number.
For any positive integer $r\geq 1$ and $k\geq 2$, the array code $\mathcal{C}(k,r,p)$ satisfies the MDS property whenever $k+r \leq p$.
\label{thm:MDS5}
\end{theorem}

\section{Efficient Decoding Method}
\label{sec:decode}

In this section, we give a decoding method based on the $\mathbf{LU}$ factorization of the Cauchy matrix in \eqref{eq:cauchy}, which is very efficient in decoding the proposed array codes. Expressing a matrix as a product of a lower triangular matrix $\mathbf{L}$ and an upper triangular matrix $\mathbf{U}$ is called an \emph{$\mathbf{LU}$ factorization}. Some results of Cauchy matrix $\mathbf{LU}$ factorization over a field can be found in \cite{boros1999fast,Calvetti1997Factorizations}. We first give an $\mathbf{LU}$ factorization of Cauchy matrix over $C_p$, and then present the efficient decoding algorithm based on the $\mathbf{LU}$ factorization.

\subsection{$\mathbf{LU}$ Factorization of Cauchy Matrix over $C_p$}

Given $2\ell$ distinct variables $a_1,a_2,\ldots,a_\ell,b_1,b_2,\ldots,b_\ell$ between 0 and $p-1$, the $\ell \times \ell$ square Cauchy matrix $\mathbf{C}(x^{a_{1:\ell}},x^{b_{1:\ell}})$ over ring $C_p$ is of the form in \eqref{eq:cauchy}. By Theorem \ref{thm:MDS5}, the matrix $\mathbf{C}(x^{a_{1:\ell}},x^{b_{1:\ell}})$ is invertible and the inverse matrix is denoted as $\mathbf{C}(x^{a_{1:\ell}},x^{b_{1:\ell}})^{-1}$. A factorization of $\mathbf{C}(x^{a_{1:\ell}},x^{b_{1:\ell}})^{-1}$ is derived, which is stated in the following theorem. No proof is given since it is similar to Theorem 3.1 in ~\cite{boros1999fast}.
\begin{theorem}
Let $x^{a_1},x^{a_2},\ldots,x^{a_\ell},x^{b_1},x^{b_2},\ldots,x^{b_\ell}$ be $2\ell$ distinct monomials in $R_p$, where $0\leq a_i,b_i<p$ for $i=1,2,\ldots,\ell$. The inverse matrix $\mathbf{C}(x^{a_{1:\ell}},x^{b_{1:\ell}})^{-1}$ can be decomposed as
\begin{equation}
\mathbf{C}(x^{a_{1:\ell}},x^{b_{1:\ell}})^{-1}=\mathbf{U}^1_\ell\mathbf{U}^2_\ell\ldots \mathbf{U}^{\ell-1}_\ell\mathbf{D}_\ell\mathbf{L}^{\ell-1}_\ell\ldots \mathbf{L}^{2}_\ell\mathbf{L}^{1}_\ell,
\label{eq:factor}
\end{equation}
where
\begin{align*}
\mathbf{L}^i_{\ell}=
\begin{bmatrix}
\mathcal{I}_i &&&\\
&\frac{1}{x^{a_{i+1}}+x^{a_1}}&&\\
&&\ddots &\\
&&& \frac{1}{x^{a_{\ell}}+x^{a_{\ell-i}}}\\
\end{bmatrix}
\begin{bmatrix}
\mathcal{I}_{i-1} &&&&\\
&e(x)&0&&\\
&x^{a_{1}}+x^{b_i}&x^{a_{i+1}}+x^{b_i}&&\\
&&\ddots &\ddots &\\
&&&x^{a_{\ell-i}}+x^{b_i}& x^{a_{\ell}}+x^{b_i}\\
\end{bmatrix},
\end{align*}
\begin{align*}
\mathbf{U}^i_\ell=
\begin{bmatrix}
\mathcal{I}_{i-1} &&&&\\
&e(x)&x^{a_{i}}+x^{b_1}&&\\
&0&x^{a_{i}}+x^{b_{i+1}}&\ddots&\\
&&&\ddots &x^{a_{i}}+x^{b_{\ell-i}}\\
&&&& x^{a_{i}}+x^{b_\ell}\\
\end{bmatrix}
\begin{bmatrix}
\mathcal{I}_{i} &&&\\
&\frac{1}{x^{b_{1}}+x^{b_{i+1}}}&&\\
&&\ddots &\\
&&& \frac{1}{x^{b_{\ell-i}}+x^{b_{\ell}}}\\
\end{bmatrix},
\end{align*}
for $i=1,2,\ldots,\ell-1$, and
\begin{equation}
\mathbf{D}_\ell=\text{diag}\{x^{a_{1}}+x^{b_{1}},x^{a_{2}}+x^{b_{2}},\ldots,x^{a_{\ell}}+x^{b_{\ell}}\}.
\label{eq:diag}
\end{equation}
\label{thm:cauchy_decomposition}
\end{theorem}

Next we give an example of the factorization. Considering  $\ell=2$, the matrix $\mathbf{C}(x^{a_{1:2}},x^{b_{1:2}})^{-1}$ can be factorized into \begin{align*}
\mathbf{U}^{1}_2\cdot \mathbf{D}_2\cdot \mathbf{L}^{1}_2=&\begin{bmatrix}
e(x) & x^{a_1}+x^{b_1}  \\
0 & x^{a_1}+x^{b_2}
\end{bmatrix}
\begin{bmatrix}
e(x) & 0  \\
0 & \frac{1}{x^{b_1}+x^{b_2}}
\end{bmatrix}
\\
&
\begin{bmatrix}
x^{a_1}+x^{b_1} & 0  \\
0 & x^{a_2}+x^{b_2}
\end{bmatrix}
\begin{bmatrix}
e(x) & 0 \\
0 & \frac{1}{x^{a_2}+x^{a_1}}
\end{bmatrix}
\begin{bmatrix}
e(x) & 0 \\
x^{a_1}+x^{b_1} & x^{a_2}+x^{b_1}
\end{bmatrix}.
\end{align*}

%

Based on the factorization in Theorem~\ref{thm:cauchy_decomposition}, we have a fast algorithm for solving a Cauchy system of linear equations over $C_p$ as that given in~\cite{boros1999fast} for a field. Given an $\ell\times \ell$ linear system in Cauchy matrix form
\begin{equation}\mathbf{C}(x^{a_{1:\ell}},x^{b_{1:\ell}})\mathbf{s}  =\mathbf{c},\label{eq:linear-equations}
\end{equation}
where $\mathbf{s}=(s_1(x), s_2(x), \ldots, s_{\ell}(x))^t$ is a column of length $\ell$ over $C_p$ and $\mathbf{c}=(c_1(x), c_2(x), \ldots, c_{\ell}(x))^t$ is a column of length $\ell$ over $C_p$. We can solve the equation for $\mathbf{s}$, given $\mathbf{C}(x^{a_{1:\ell}},x^{b_{1:\ell}})$ and $\mathbf{c}$, by computing
\begin{equation}
\mathbf{U}^1_\ell\mathbf{U}^2_\ell\ldots \mathbf{U}^{\ell-1}_\ell\mathbf{D}_\ell\mathbf{L}^{\ell-1}_\ell\ldots \mathbf{L}^{2}_\ell\mathbf{L}^{1}_\ell\mathbf{c}.
\label{Cauchy-factor}
\end{equation}
The pseudocode is stated in Algorithm \ref{alg:A1}.

\begin{algorithm}[!htb]
\caption{Solving a Cauchy linear system over $C_p$.}
\label{alg:A1}
{\bf Inputs:}

\quad Positive integer $\ell$, prime number $p>2$, the values of $\mathbf{c}=(c_1(x), c_2(x), \ldots, c_{\ell}(x))^t$, $a_1,a_2,\ldots, a_\ell$ and $b_1,b_2,\ldots, b_\ell$.

{\bf Outputs:}

\quad The values of $\mathbf{s}=(s_1(x), s_2(x), \ldots, s_{\ell}(x))^t$.

\begin{algorithmic}[1]

\REQUIRE  All  $2\ell$ distinct non-negative integers $a_1,\ldots, a_\ell,b_1,\ldots, b_\ell$ are less than $p$.
\FOR {$i=1,2,\ldots,\ell$}
\STATE $s_i(x)=c_i(x)$.
\ENDFOR

\FOR {$i=1,2,\ldots,\ell-1$}
\FOR {$j=i+1,i+2,\ldots,\ell$}
\STATE  $s_j(x)=(x^{a_{j-i}}+x^{b_{i}})s_{j-1}(x)+(x^{a_{j}}+x^{b_{i}})s_j(x)$.
\ENDFOR
\FOR {$j=i+1,i+2,\ldots,\ell$}
\STATE  $s_j(x)=\frac{1}{x^{a_{j}}+x^{a_{j-i}}} s_j(x)$.
\ENDFOR
\ENDFOR

\FOR {$i=1,2,\ldots,\ell$}
\STATE $s_i(x)=(x^{a_{i}}+x^{b_{i}})s_i(x)$.
\ENDFOR

\FOR {$i=\ell-1,\ell-2,\ldots,1$}
\FOR {$j=i+1,i+2,\ldots,\ell$}
\STATE $s_j(x)=\frac{1}{x^{b_{j-i}}+x^{b_{j}}}s_{j}(x)$.
\ENDFOR
\FOR {$j=1,2,\ldots,\ell$}
\IF {$j-i=0$}
\STATE $s_j(x)=s_j(x)+(x^{a_{i}}+x^{b_{j-i+1}})s_{j+1}(x)$.
\ENDIF
\IF {$\ell-1 \geq j-i \geq 1$}
\STATE $s_j(x)=(x^{a_{i}}+x^{b_{j}})s_j(x)+(x^{a_{i}}+x^{b_{j-i+1}})s_{j+1}(x)$.
\ENDIF
\ENDFOR
\STATE $s_\ell(x)=(x^{a_{i}}+x^{b_{\ell}})s_\ell(x)$.
\ENDFOR

\end{algorithmic}
\end{algorithm}

\subsection{Decoding Algorithm of Erasures}
We now describe the decoding procedure of any $\rho\leq r$ erasures for the array codes $\mathcal{C}(k,r,p)$. Suppose that $\gamma$ information columns $a_{1},a_{2},\ldots,a_{\gamma}$ and $\delta$ parity columns $b_{1},b_{2},\ldots,b_{\delta}$ erased with $0\leq a_{1}<a_{2}<\ldots<a_{\gamma}\leq k-1$ and $0\leq b_{1}<b_{2}<\ldots<b_{\delta}\leq r-1$, where $k\geq \gamma \geq 0$, $r\geq \delta \geq 0$ and $\gamma+\delta=\rho\leq r$. Let $$\mathcal{A}:=\{0,1,\ldots,k-1\}\setminus \{a_{1},a_{2},\ldots,a_{\gamma}\}$$
be set of the indices of the available information columns, and let
$$\mathcal{B}:=\{0,1,\ldots,r-1\}\setminus \{b_{1},b_{2},\ldots,b_{\delta}\}$$
be set of the indices of the available parity columns.

We want to first recover the lost information columns by reading $k-\gamma$ information columns with indices $i_1,i_2,\ldots, i_{k-\gamma}\in\mathcal{A}$, and $\gamma$ parity columns with indices $\ell_1, \ell_2,\ldots,\ell_\gamma \in \mathcal{B}$, and then recover the failure parity column by multiplying the corresponding encoding vector and the $k$ data polynomials.

For $\tau=1,2,\ldots,k-\gamma$, we add the extra parity-check bit for information column $i_{\tau}$ to obtain the data polynomial
$$s_{i_\tau}(x)=s_{0,i_\tau}+s_{1,i_\tau}x+\cdots+s_{p-2,i_\tau}x^{p-2}+(\sum_{j=0}^{p-2}s_{j,i_\tau})x^{p-1}.$$
For $h=1,2,\ldots,\gamma$, since the coded polynomial $c_{\ell_h}(x)\in C_p$, we have
$$c_{\ell_h}(x)=c_{0,\ell_h}+c_{1,\ell_h}x+\cdots+c_{p-2,\ell_h}x^{p-2}+(\sum_{j=0}^{p-2}c_{j,\ell_h})x^{p-1}.$$
Let $p_{\ell_1}(x),p_{\ell_2}(x),\ldots,p_{\ell_\gamma}(x)$ be the polynomials by subtracting
the chosen $k-\gamma$ data polynomials $s_{i_1}(x),s_{i_2}(x),\ldots,s_{i_{k-\gamma}}(x)$ from $\gamma$ coded polynomials $c_{\ell_1}(x),c_{\ell_2}(x),\ldots,c_{\ell_\gamma}(x)$, i.e.,
\begin{equation}
p_{\ell_h}(x)\triangleq c_{\ell_h}(x)+\sum_{j=1}^{k-r}\frac{1}{x^{\ell_h}+x^{\ell_h+i_j+r-1}}s_{i_j}(x),
\label{eq:px}
\end{equation}
for $h=1,2,\ldots,\gamma$.
We can obtain the $\gamma$ information erasures by solving the following system of linear equations
\begin{equation}
\begin{bmatrix}
\frac{1}{x^{\ell_1}+x^{a_1+r}} & \frac{1}{x^{\ell_1}+x^{a_2+r}} & \cdots & \frac{1}{x^{\ell_1}+x^{a_\gamma+r}} \\
\frac{1}{x^{\ell_2}+x^{a_1+r}} & \frac{1}{x^{\ell_2}+x^{a_2+r}} & \cdots & \frac{1}{x^{\ell_2}+x^{a_\gamma+r}} \\
\vdots   & \vdots   & \ddots & \vdots \\
\frac{1}{x^{\ell_\gamma}+x^{a_1+r}} & \frac{1}{x^{\ell_\gamma}+x^{a_2+r}} & \cdots & \frac{1}{x^{\ell_\gamma}+x^{a_\gamma+r}} \\
\end{bmatrix}
\begin{bmatrix}
s_{a_{1}}(x) \\ s_{a_{2}}(x)\\ \vdots \\ s_{a_{\gamma}}(x)
\end{bmatrix}
=
\begin{bmatrix}
p_{\ell_1}(x) \\ p_{\ell_2}(x)\\ \vdots \\ p_{\ell_\gamma}(x)
\end{bmatrix}.\label{eq:decoding-linear-system}
\end{equation}
The above system of linear equations is with the form of \eqref{eq:linear-equations} such that
Algorithm \ref{alg:A1} can be applied to obtain  the $\gamma$ failure data polynomials. Then we can recover the $\delta$ coded polynomials by multiplying the corresponding encoding vectors and $k$ data polynomials.

\subsection{Computation Complexity of Linear System in Cauchy Matrix}
\subsubsection{Algorithm for division}

 In computing the coded polynomial in~\eqref{eq:en} and in Algorithm~\ref{alg:A1}, we should compute many divisions of the form $\frac{s(x)}{x^t+x^{t+b}}$, where $s(x)\in C_p$, $b$ and $t$ are non-negative integers such that $b+t<p$. Let's first consider the calculation of
\begin{equation}
 \frac{s(x)}{x^t+x^{t+b}}= c(x) \mod (1+x^p),
\label{division}
\end{equation}
where $s(x)\in C_p$ and $c(x)\in R_p$. If $c(x)\not\in C_p$, we can take $c(x)+h(x)$, which is in $C_p$, instead. Later we will show that the above step is not necessary in encoding and decoding processes if we allow some coded polynomials to be of form $c(x)+h(x)$.
The following lemma demonstrates an efficient method to compute \eqref{division}.


\begin{algorithm}[!htb]
\caption{Solving the division given in  \eqref{division}.}
\label{alg:division}
{\bf Inputs:}

\quad Non-negative integers $b,t$ and $s_0,s_1,\ldots,s_{p-1}$, where $s_i\in\{0,1\}$ for $i=0,1,\ldots,p-1$.

{\bf Outputs:}

\quad $c_0,c_1,\ldots,c_{p-1}$, where $c_i\in\{0,1\}$ for $i=0,1,\ldots,p-1$.

\begin{algorithmic}[1]

\REQUIRE  Both $b+t<p$  and  $s_0+s_1+\ldots +s_{p-1}=0$ hold.
\STATE $c_{p-1}=0$.
\STATE $c_{p-b-1}=s_{(t-1)\bmod p}$.
\STATE $c_{b-1}=s_{(t+b-1)\bmod p}$.
\FOR {$i=2,3,\ldots,p-2$}
\STATE $c_{(p-ib-1)\bmod p}=s_{(t-(i-1)b-1)\bmod p}+c_{p-(i-1)b-1}$.
\ENDFOR
\IF {$\sum_{i=0}^{p-1}c_i\neq 0$}
 \FOR {$i=0,1,\ldots, p-1$} \STATE $c_i=c_i+1$. \ENDFOR
\ENDIF
\end{algorithmic}
*Steps 6, 7, and 8 are deleted in simplified version.
\end{algorithm}

\begin{lemma}
 The coefficients of $c(x)$ in \eqref{division} can be computed by Algorithm \ref{alg:division}, where $t,b\ge 0$, $0< b+t<p$, $s(x)=\sum_{i=0}^{p-1}s_ix^i\in C_p$, and $c(x)=\sum_{i=0}^{p-1}c_ix^i\in C_p$.
\label{lm:inverse2}
\end{lemma}
\begin{proof}
By \eqref{division} and Lemma~\ref{lm:inverse}, we have
\begin{equation}
\label{eq:division-2}
s(x)\left(x^{p-t}(1+x^{2b}+x^{4b}+\cdots + x^{(p-1)b})\right)
= c(x) \mod (1+x^p).
\end{equation}
Multiplied  by $x^t+x^{t+b}$, \eqref{eq:division-2} becomes
\begin{equation}
\label{eq:division-3}
s(x) (1+h(x))
= c(x)(x^t+x^{t+b}) \mod (1+x^p).
\end{equation}
By \eqref{zero}, \eqref{eq:division-3} is equivalent to
\begin{equation}
\label{eq:division-4}
s(x)
= c(x)(x^t+x^{t+b}) \mod (1+x^p).
\end{equation}
Then, the coefficients of $s(x)$ and $c(x)$ satisfy
\begin{eqnarray}\label{eq:coefficient}
s_{t}&=&c_0+c_{p-b},\nonumber\\
s_{t+1}&=&c_1+c_{p-b+1},\nonumber\\
s_{(t+2)\bmod p}&=&c_2+c_{p-b+2},\nonumber\\
&&\text{  }\text{  }\text{  }\text{  }\text{  }\vdots\\
s_{(t-2)\bmod p}&=&c_{p-2}+c_{p-b-2},\nonumber\\
s_{(t-1)\bmod p}&=&c_{p-1}+c_{p-b-1}.\nonumber
\end{eqnarray}
Recall that, given $s(x)$, $t$ and $b$,  there are two polynomials $c(x)=\sum_{i=0}^{p-1}c_ix^i$ and $c(x)+h(x)=\sum_{i=0}^{p-1}(c_i+1)x^i$ such that
$$\Big( \sum_{i=0}^{p-1}c_ix^i\Big)(x^t+x^{t+b})=\Big(\sum_{i=0}^{p-1}(c_i+1)x^i\Big)(x^t+x^{t+b})=s(x).$$
We can choose one coefficient $c_i$ of $c(x)$ to be zero, and all the other coefficients can be computed iteratively. Specifically, in Algorithm \ref{alg:division}, we let $c_{p-1}=0$. Then we obtain $c_{p-b-1}=s_{(t-1)\bmod p}$ and $c_{b-1}=s_{(t+b-1)\bmod p}$.
 Substituting $c_{p-b-1}$ into the corresponding equation in \eqref{eq:coefficient},  we have
$$c_{(p-2b-1)\bmod p}=s_{(t-b-1)\bmod p}+c_{p-b-1}.$$
In general, we have
$$c_{(p-ib-1)\bmod p}=s_{(t-(i-1)b-1)\bmod p}+c_{(p-(i-1)b-1)\bmod p}$$ for $2\le i\le p-2$. Note that each  coefficient can be calculated iteratively with at most one XOR operation involved. Next we need to prove that $$\{(p-ib-1)\bmod p|1\le i\le p-2\}=\{0,1,2,\ldots,p-2\}.$$ First we prove that if $i\neq j$, then $(p-ib-1)\bmod p\neq (p-jb-1)\bmod p$. Assume that $(p-ib-1)\bmod p= (p-jb-1)\bmod p$ and $1\le j<i\le p-2$. Then there exists an integer $\ell$ such that
$$p-jb-1=\ell p+p-ib-1.$$ The above equation can be further reduced to
$$(i-j)b=\ell p.$$ Since either $b=1$ or $b\not| p$, we have $(i-j)|p$. However, this is impossible due to the fact that  $1\le j<i\le p-2$. Similarly, we can prove that, for $1\le i\le p-2$,
$$p-ib-1 \bmod p\neq p-1.$$ Hence, $\{(p-ib-1)\bmod p|1\le i\le p-2\}=\{0,1,2,\ldots,p-2\}$. Finally, if $\sum_{i=0}^{p-1}c_i\neq 0$, then $c(x)\not\in C_p$; however, $c(x)+h(x)\in C_p$.
\end{proof}

\subsubsection{Simplified algorithm for division}

Next, we prove that Steps 6, 7 and 8  in Algorithm~\ref{alg:division} are not necessary. Hence, the computation complexity of Algorithm~\ref{alg:division} can be reduced drastically. We name Algorithm~\ref{alg:division} without Steps 6, 7 and 8 as simplified Algorithm~\ref{alg:division}.

Recall that, after dropping Steps $6-8$ in Algorithm~\ref{alg:division}, the output of the algorithm might be $c(x)+h(x)$ instead of $c(x)$; however, we will show that the data polynomials can be recovered after performing the proposed decoding algorithm no matter which algorithm is performed.
\begin{theorem}
\label{equivalent-code}
The proposed decoding 	algorithm outputs the same data polynomials no matter Algorithm~\ref{alg:division} or simplified Algorithm~\ref{alg:division} is performed.
\end{theorem}
\begin{proof}
According to Algorithm~\ref{alg:A1}	, there are two steps (Step 7 and Step 12) in it involve (simplified) Algorithm~\ref{alg:division}. In addition, after encoding, $c(x)$ might become $c(x)+h(x)$ when we applied simplified Algorithm~\ref{alg:division} for encoding.
When applying simplified Algorithm~\ref{alg:division} in the encoding process, the coded polynomials might be $c(x)+h(x)$ instead of $c(x)$. Hence, the input of Algorithm~\ref{alg:A1} becomes  $c_1(x)+a_1h(x), c_2(x)+a_2h(x), \ldots, c_\ell(x)+a_\ell h(x)$, where $a_i\in \{0,1\}$ for $1\le i\le \ell$.
Note that since $h(x)$ is the check polynomial of $C_p$, from \eqref{zero}, we have
\begin{equation}
s(x)(c(x)+h(x))=s(x)c(x) \mod (1+x^p)
\label{eq:check}
\end{equation}
$\forall s(x)\in C_p$ and $c(x)\in R_p$. Hence, after performing Step 5 in Algorithm~\ref{alg:A1}, $s_j\in C_p$ for $2\le j\le \ell$. However, after performing Step 7, $s_j(x)$ might become $s_j(x)+h(x)$ for $2\le j\le \ell$ due to performing simplified Algorithm~\ref{alg:division}. Again, after performing Step 9, the effect of $h(x)$ has eliminated according to \eqref{eq:check}. Similar argument can be applied for performing Step 12, Step 15 (or Step 17), and Step 18. Hence, we can conclude that the output of Algorithm~\ref{alg:A1} becomes the same no mater Algorithm~\ref{alg:division} or simplified Algorithm~\ref{alg:division} is applied.
\end{proof}

\subsubsection{Computation complexity}

Recall that the coded polynomial $c_j(x)$, for $j=0,1,\ldots, r-1$, is computed by
\begin{equation}
c_j(x)=\frac{1}{x^j+x^r}s_0(x)+\frac{1}{x^j+x^{r+1}}s_1(x)+\cdots+\frac{1}{x^j+x^{k+r-1}}s_{k-1}(x).
\label{eq:coded}
\end{equation}
Next we determine the computation complexity of the proposed decoding algorithm.
With Lemma \ref{lm:inverse2}, we have that the last coefficient of polynomial $\frac{1}{x^j+x^{r+\ell}}s_{\ell}(x)$ is equal to 0, $\ell=0,1,\ldots, k-1$. Therefore, the last coefficient $c_{p-1,j}$ of the coded polynomial $c_j(x)$ is equal to 0. Hence, there are $p-3$ XORs involved in computing $\frac{1}{x^j+x^{r+\ell}}s_{\ell}(x)$ by simplified Algorithm~\ref{alg:division}.


Note that, in Algorithm~\ref{alg:A1}, we only need to compute three different operations: (i) multiplication of $c_i(x)$ and $x^{a_i}$, (ii) division of the form $\frac{c_i(x)}{x^{a_j}+x^{b_{j-i}}}$, (iii) addition between $c_i(x)$ and $c_j(x)$. Hence, in Algorithm \ref{alg:A1},  there are total $4\ell(\ell-1)+2\ell$ multiplications of the first type, $\ell(\ell-1)$ divisions of the second type and $\ell+3\ell(\ell-1)$ additions.

The multiplication of $x^i$ and a polynomial $s(x)$ over $R_p$ can be obtained by cyclically shifting the polynomial $s(x)$ by $i$ bits, which takes no XORs. The second operation over $R_p$ requires $p-3$ XORs with performing simplified Algorithm~\ref{alg:division}. One addition needs $p$ XORs.
In Algorithm~\ref{alg:A1}, Steps 4 and 5 are the computation of the right matrix of $\mathbf{L}^i_\ell$ and the column vector $\mathbf{c}$ of length $\ell$ with each component being a polynomial in $R_p$, of which the complexity is at most $3(\ell-i)p$ XORs. In the resultant column vector $\mathbf{c}$, the first $i$ components are in $C_p$ and the last $\ell-i$ components are in $C_p$. Steps 6 to 7  calculate the left matrix of $\mathbf{L}^i_\ell$ and the above resultant column vector $\mathbf{c}$. As the last $\ell-i$ components are in $C_p$, all the divisions of the form $\frac{c_i(x)}{x^j+x^{j-i}}$ can be computed by simplified Algorithm~\ref{alg:division}, which takes $(\ell-i)(p-3)$ XORs. Recall that the last bit of polynomial $\frac{c_i(x)}{x^j+x^{j-i}}$ is zero, and the multiplication of $x^i+x^j$ and $\frac{c_i(x)}{x^j+x^{j-i}}$ thus requires $p-2$ XORs. Therefore, the total number of XORs involved in Steps 4 to 7 are
$$\underbrace{3p\ell(\ell-1)/2}_{\text{Steps 4 to 5}}+\underbrace{(p-3)\ell(\ell-1)/2}_{\text{Steps 6 to 7}}.$$
Steps 8 and 9 compute the multiplication of diagonal matrix $\mathbf{D}_\ell$ and the above resultant column vector, where the number of XORs involved are $p+(p-2)(\ell-1)$.
Steps 11 and 12 compute multiplication of  the right matrix of $\mathbf{U}^i_\ell$ and the above column vector, where $(p-3)\ell(\ell-1)/2$ XORs are required. Steps 13 to 18 calculate multiplication of  the left matrix $\mathbf{U}^i_\ell$ and the above column vector, where $(2(p-2)+p)\ell(\ell-1)/2$ XORs are needed. Therefore, the total number of XORs involved in Algorithm \ref{alg:A1} over $R_p$ is at most
\begin{align*}
&\underbrace{3p\ell(\ell-1)/2}_{\text{Steps 4 to 5}}+\underbrace{(p-3)\ell(\ell-1)/2}_{\text{Steps 6 to 7}}+\underbrace{p+(p-2)(\ell-1)}_{\text{Steps 8 to 9}}+\\
&\underbrace{(p-3)\ell(\ell-1)/2}_{\text{Steps 11 to 12}}+\underbrace{(2(p-2)+p)\ell(\ell-1)/2}_{\text{Steps 13 to 18}}\\
&=4\ell^2p-3\ell p-5\ell^2+3\ell+2.
\end{align*}

Adding overall parity-checks to $k-\gamma$ data polynomials takes $(k-\gamma)(p-2)$ XORs. Computing $\gamma$ polynomials
in \eqref{eq:px} requires $\gamma((k-\gamma)(p-3)+(k-\gamma)(p-1))=\gamma(k-\gamma)(2p-4)$ XORs. The number of XORs
involved in solving the $\gamma\times \gamma$ Cauchy system is $4\gamma^2p-3\gamma p-5\gamma^2+3\gamma+2$. In recovering
the $\delta$ parity columns, there are $\delta(k(p-3)+(k-1)(p-1))$ XORs involved. Therefore the decoding complexity of
recovering $\gamma$ information erasures and $\delta$ parity erasures is
$$(k-\gamma)(p-2)+\gamma(k-\gamma)(2p-4)+4\gamma^2p-3\gamma p-5\gamma^2+3\gamma+2+\delta(k(p-3)+(k-1)(p-1)) \text{ XORs}.$$
When $\delta=0$, i.e., only information column fails, the decoding complexity is
$$(k-\gamma)(p-2)+\gamma(k-\gamma)(2p-4)+4\gamma^2p-3\gamma p-5\gamma^2+3\gamma+2 \text{ XORs}.$$



\subsection{Example of Cauchy Array Codes}
Consider  Cauchy array codes $\mathcal{C}(3,2,5)$, where $k=2,r=2,p=5$.
There are two data polynomials $s_i(x)=s_{0,i}+s_{1,i}x+s_{2,i}x^2+s_{3,i}x^3+(s_{0,i}+s_{1,i}+s_{2,i}+s_{3,i})x^4$, for $i=0,1$. Two coded polynomials $c_0(x)=c_{0,0}+c_{1,0}x+c_{2,0}x^2+c_{3,0}x^3+c_{4,0}x^4$ and $c_1(x)=c_{0,1}+c_{1,1}x+c_{2,1}x^2+c_{3,1}x^3+c_{4,1}x^4$ are computed by
\begin{align*}
c_0(x)\triangleq&\frac{1}{1+x^2} s_0(x)+ \frac{1}{1+x^{3}}s_1(x),\\
c_1(x)\triangleq &\frac{1}{x+x^2} s_0(x)+ \frac{1}{x+x^{3}}s_1(x).
\end{align*}
$\frac{s(x)}{1+x^b}$ can be solved with 2 XORs by the simplified Algorithm~\ref{alg:division}. For  $c(x)=\frac{s_0(x)}{1+x^2}$, we have $s_{0,0}=c_0+c_3$, $s_{1,0}=c_1+c_4$, $s_{2,0}=c_2+c_0$, $s_{3,0}=c_3+c_1$ and $s_{4,0}=c_4+c_2$. First we set $c_4=0$, we have $c_1=s_{1,0}$ and $c_2=s_{4,0}$. Then we can compute $c_0=s_{2,0}+s_{4,0}$ from $c_0=s_{2,0}+c_2$, and $c_3=c_0+s_{0,0}=s_{2,0}+s_{4,0}+s_{0,0}$. The total number of XORs involved in computing $\frac{s_0(x)}{1+x^2}$ is 2.

\begin{table*}[t]
\caption{The array code $\mathcal{C}(3,2,5)$.}
\label{table:325}
\[
\begin{array}{|c|c|c|c|c|c|c|} \hline
\text{Disk 0}&\text{Disk 1}&\text{Disk 2}&\text{Disk 3}   \\ \hline \hline
s_{0,0} & s_{0,1} & (s_{2,0}+s_{4,0})+(s_{1,1}+s_{3,1}+s_{4,1}) & (s_{0,0}+s_{1,0}+s_{3,0})+(s_{0,1}+s_{3,1}) \\ \hline
s_{1,0} & s_{1,1} & s_{1,0}+s_{4,1} & s_{1,0}+s_{2,1} \\ \hline
s_{2,0} & s_{2,1} & s_{4,0}+s_{2,1} & s_{4,0}+s_{0,1} \\ \hline
s_{3,0} & s_{3,1} & (s_{0,0}+s_{2,0}+s_{4,0})+(s_{1,1}+s_{4,1}) & (s_{1,0}+s_{3,0})+(s_{0,1}+s_{1,1}+s_{3,1}) \\ \hline \hline
s_{4,0} & s_{4,1} & 0 & 0 \\ \hline
\end{array}
\]
\vspace{-0.6cm}
\end{table*}
The code given in the above example is shown in Table~\ref{table:325}. The last row of the array code in Table~\ref{table:325} does not need to be stored, as the last bit of each information column is the parity-check bit of the first $p-1$ bits and the last bit of each parity column is zero.

Assume that two data polynomials $s_0(x),s_1(x)$ are $1+x$ and $x+x^3$ respectively, then the two coded polynomials are computed as $c_0(x)=x$ and $c_1(x)=x+x^2+x^3$.

By Theorem~\ref{thm:cauchy_decomposition}, the inverse matrix of the $2\times 2$ Cauchy matrix can be factorized into
\begin{align*}
&\mathbf{U}^1_2\cdot \mathbf{D}_2\cdot \mathbf{L}^1_2=\\
&
\begin{bmatrix}
e(x) & 1+x^2  \\
0 & 1+x^3
\end{bmatrix}
\begin{bmatrix}
e(x) & 0  \\
0 & \frac{1}{x^2+x^3}
\end{bmatrix}\cdot
\begin{bmatrix}
1+x^2 & 0  \\
0 & x+x^3
\end{bmatrix}\cdot
\begin{bmatrix}
e(x) & 0  \\
0 & \frac{1}{1+x}
\end{bmatrix}
\begin{bmatrix}
e(x) & 0  \\
1+x^2 & x+x^2
\end{bmatrix}.
\end{align*}
We can check that the two data polynomials can be recovered by
\begin{eqnarray*}
&&\mathbf{U}^1_2\cdot \mathbf{D}_2\cdot \mathbf{L}^1_2 \cdot \begin{bmatrix}x \\x+x^2+x^3\end{bmatrix}\\
&=&\mathbf{U}^1_2\cdot \mathbf{D}_2\cdot \begin{bmatrix}
e(x) & 0  \\
0 & \frac{1}{1+x}
\end{bmatrix}
\begin{bmatrix}
e(x) & 0  \\
1+x^2 & x+x^2
\end{bmatrix}\begin{bmatrix}x \\x+x^2+x^3\end{bmatrix} \\
&=&\mathbf{U}^1_2\cdot \mathbf{D}_2\cdot \begin{bmatrix}
e(x) & 0  \\
0 & \frac{1}{1+x}
\end{bmatrix}\begin{bmatrix}
x  \\
1+x+x^2+x^3
\end{bmatrix}\\
&=&\mathbf{U}^1_2\cdot \begin{bmatrix}
1+x^2 & 0  \\
0 & x+x^3
\end{bmatrix}\cdot \begin{bmatrix}
x  \\
1+x^2
\end{bmatrix}\\
&=&\begin{bmatrix}
e(x) & 1+x^2  \\
0 & 1+x^3
\end{bmatrix}
\begin{bmatrix}
e(x) & 0  \\
0 & \frac{1}{x^2+x^3}
\end{bmatrix}\cdot \begin{bmatrix}
x+x^3  \\
1+x
\end{bmatrix}\\
&=&\begin{bmatrix}
e(x) & 1+x^2  \\
0 & 1+x^3
\end{bmatrix}
\cdot \begin{bmatrix}
x+x^3  \\
x^3
\end{bmatrix}\\
&=&\begin{bmatrix}
1+x  \\
x+x^3
\end{bmatrix},
\end{eqnarray*}
with 32 XORs involved.

\section{Performance Comparisons}
\label{sec:comparison}
In this section, we evaluate the encoding/decoding complexities for the proposed $\mathcal{C}(k,r,p)$ as well as
other existing Cauchy family array codes, such as Rabin-like code \cite{feng2005new2}, Circulant Cauchy codes
\cite{schindelhauer2013maximum} and CRS code \cite{Blomer1999An}, which is widely employed in many practical
distributed storage systems such as Facebook data centers \cite{XorbasVLDB}.


CRS code is constructed by Cauchy matrices \cite{plank1997tutorial}. It uses projections that convert the
operations of finite filed multiplication into XORs. This leads to reduction on coding complexity because the
standard RS algorithm \cite{plank1997tutorial} consumes most of the time over finite field multiplications.
As the state-of-the-art works in correcting 4 or more erasures, Rabin-like code, Circulant Cauchy codes and CRS code
are used as main comparison to the proposed codes. Note that the coding algorithm of CRS code
involves Cauchy matrices, and it is hard to calculate the exact number of ones in the Cauchy matrices. We run
simulations for CRS code and record the average numbers from simulations to estimate the encoding/decoding complexity.

We determine the \emph{normalized encoding complexity} as the ratio of the encoding complexity to the number of
information bits, and \emph{normalized decoding complexity} as the ratio of the decoding complexity to the
number of information bits.

\subsection{Encoding Complexity}

In the  $p-1\times (k+r)$ array of code $\mathcal{C}(k,r,p)$, there are $k$ information columns and $r$ parity columns.
First, we should compute a parity-check bit for each information column to obtain $k$ data polynomials, with $k(p-2)$
XORs being involved. Second, we need to compute $r$ coded polynomials by \eqref{eq:en}. There are $p-3$ XORs
required to compute a division of form $x^t+x^{t+b}$ by simplified Algorithm~\ref{alg:division}. Each coded polynomial is generated by
computing $k$ divisions of form $1+x^b$ and $k-1$ additions. As the last coefficient is zero (by Lemma \ref{lm:inverse2}),
the $k-1$ additions takes $(k-1)(p-1)$ XORs. Therefore, $k(p-3)+(k-1)(p-1)$ XORs are required to obtain a coded
polynomial. The total number of XORs required for construction $r$ parity columns are $k(p-2)+r(2kp-4k-p+1)$,
and the normalized encoding complexity is
$$\frac{k(p-2)+r(2kp-4k-p+1)}{k(p-1)}.$$
The normalized encoding complexity of Rabin-like code and Circulant Cauchy codes is the same, which is
$3r-2+\frac{k-r}{k(p-1)}$ \cite{schindelhauer2013maximum}.

\begin{figure}[tbh]
\subfloat[$r$ = 4.]{\centering{}\includegraphics[width=0.5\textwidth]{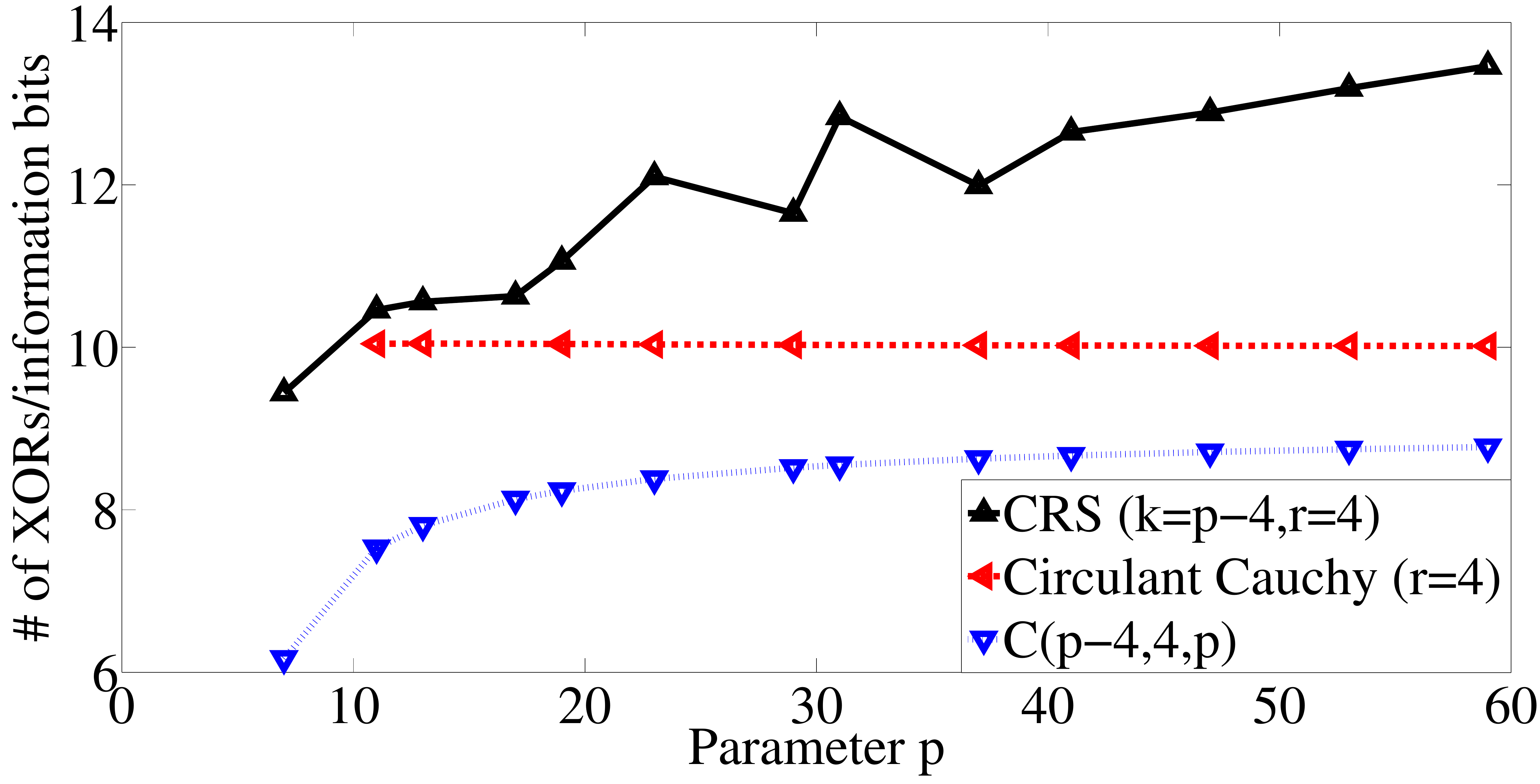}}
\subfloat[$r$ = 5.]{\begin{centering}
\includegraphics[width=0.5\columnwidth]{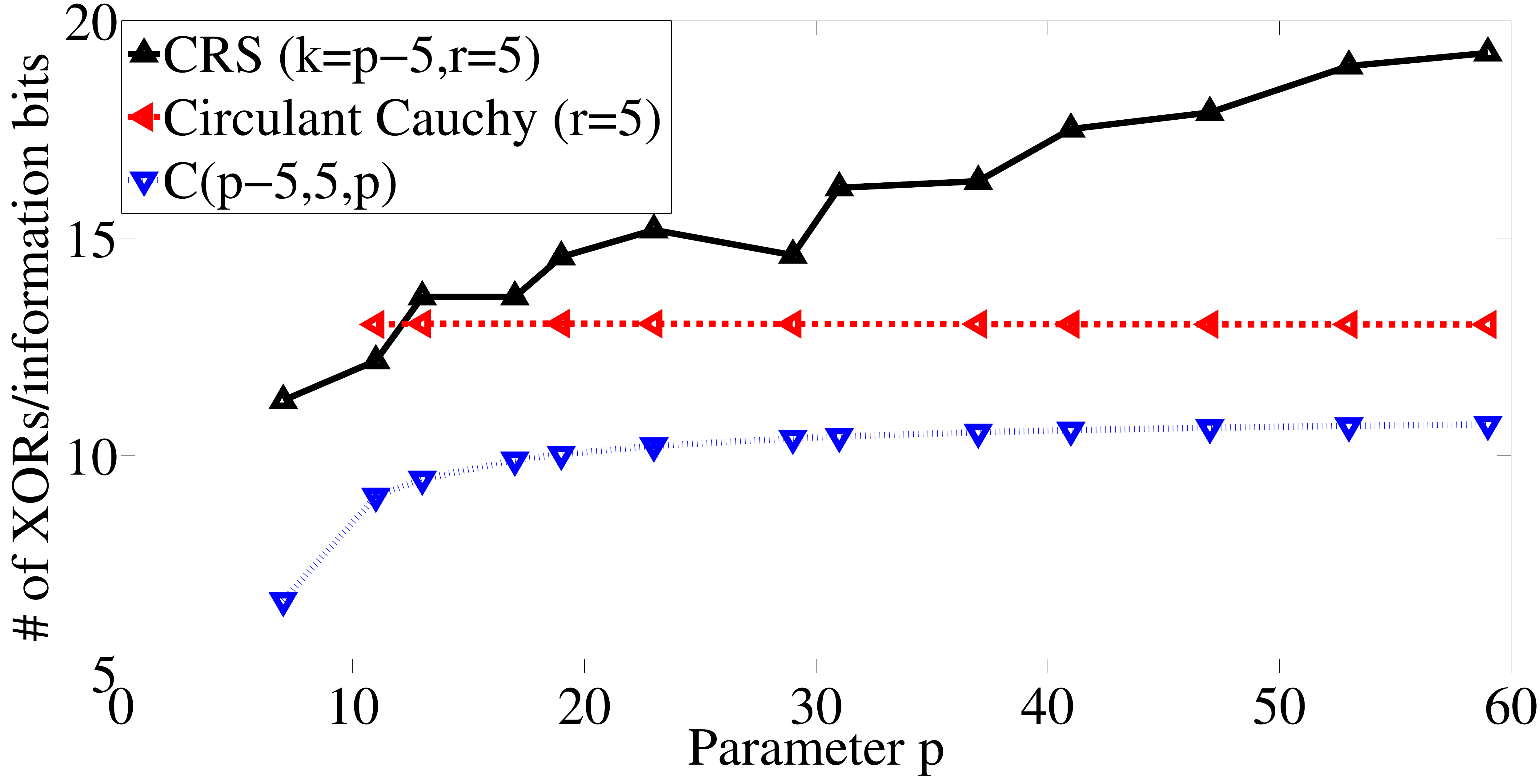}
\par\end{centering}

}\caption{The normalized encoding complexity. The complexity has been normalized to the
number of information bits.}
\vspace{-0.4cm}
\label{en:complexity}
\end{figure}


For fair comparison, we set $k=p-r$ for the three codes, we can thus have
the normalized encoding complexity of the proposed $\mathcal{C}(p-r,r,p)$ as
$$\frac{(p-r)(p-2)+r(2p(p-r)-4(p-r)-p+1)}{(p-r)(p-1)}.$$
The normalized encoding complexities of Circulant Cauchy codes, CRS code and $\mathcal{C}(p-r,r,p)$ for $r=4$ and $r=5$
are shown in Fig. \ref{en:complexity}.
For all the values of parameter $p$, the encoding complexity of $\mathcal{C}(p-r,r,p)$ is less than those
of Circulant Cauchy codes and CRS codes. Note that the difference between the proposed code and others becomes larger when $r$ increases. When $r=4$,  the reduction on  the encoding complexity of $\mathcal{C}(p-4,4,p)$ over
Circulant Cauchy codes and CRS codes are 12.5\%-38.0\% and 23.6\%-34.9\%, respectively. When  $r=5$,
they increases to  17.7\%-47.8\%
and 25.6\%-44.4\%, respectively.

\subsection{Decoding Complexity}
In the following, we evaluate the decoding complexity of the proposed array codes $\mathcal{C}(k,r,p)$,
CRS codes and Circulant Cauchy codes. If no information column fails, then  the decoding procedure of parity column failure can be viewed as
a special case of the encoding procedure. Hence,  we only consider the case with at least one information column fail.

We let $k=p-r$ for the three codes, and we have
the normalized decoding complexity of the proposed $\mathcal{C}(p-r,r,p)$ as
$$\frac{(p-2r)(p-2)+r(p-2r)(2p-4)+4r^2p-3r p-5r^2+3r+2}{(p-r)(p-1)}.$$
The authors in \cite{schindelhauer2013maximum} gave the normalized decoding complexity of
Circulant Cauchy codes as
$\frac{3rp(p-r)+6r^2p}{(p-r)(p-1)}$.

\begin{figure}[tbh]
\subfloat[$r$ = 4.]{\centering{}\includegraphics[width=0.5\textwidth]{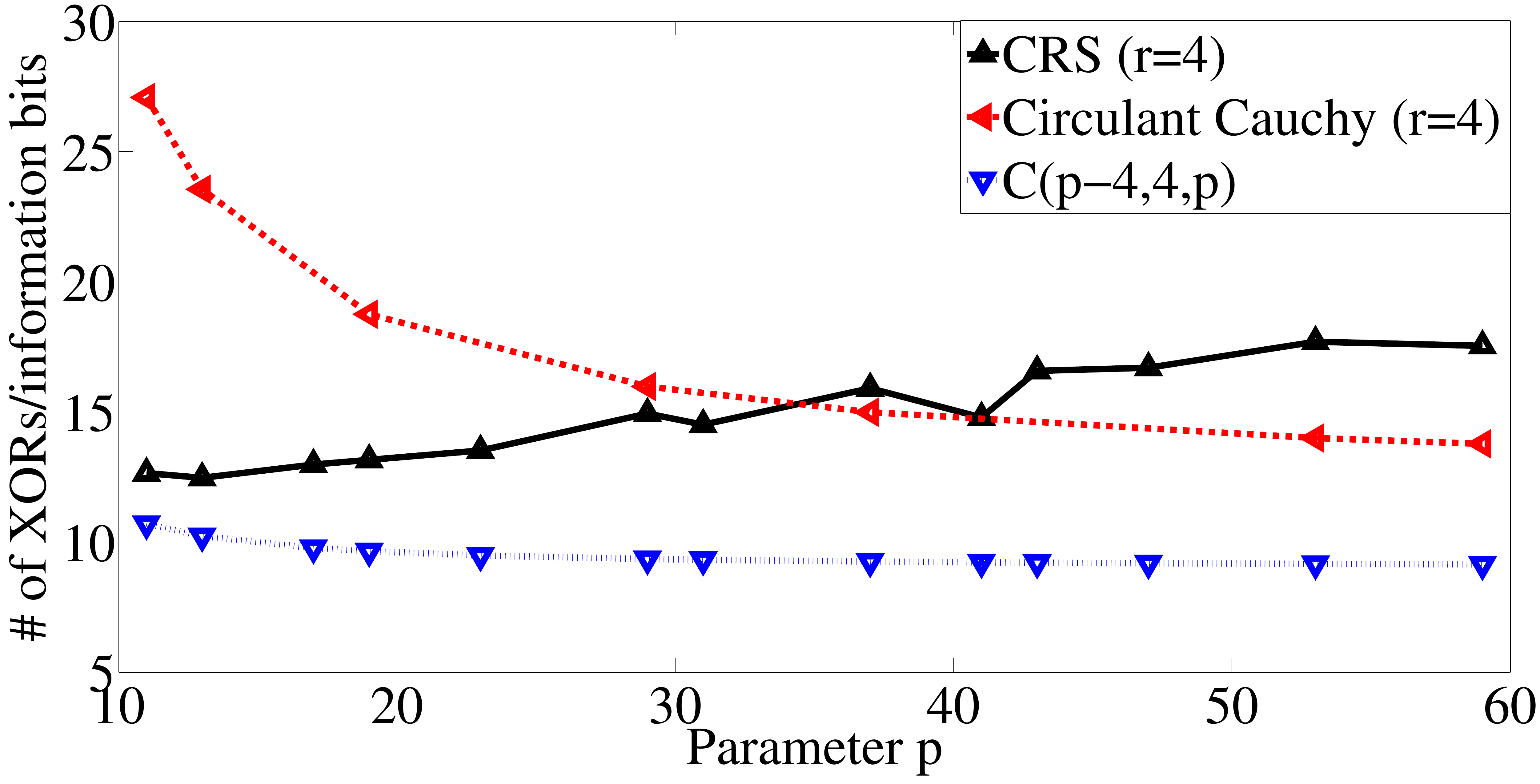}}
\subfloat[$r$ = 5.]{\begin{centering}
\includegraphics[width=0.5\columnwidth]{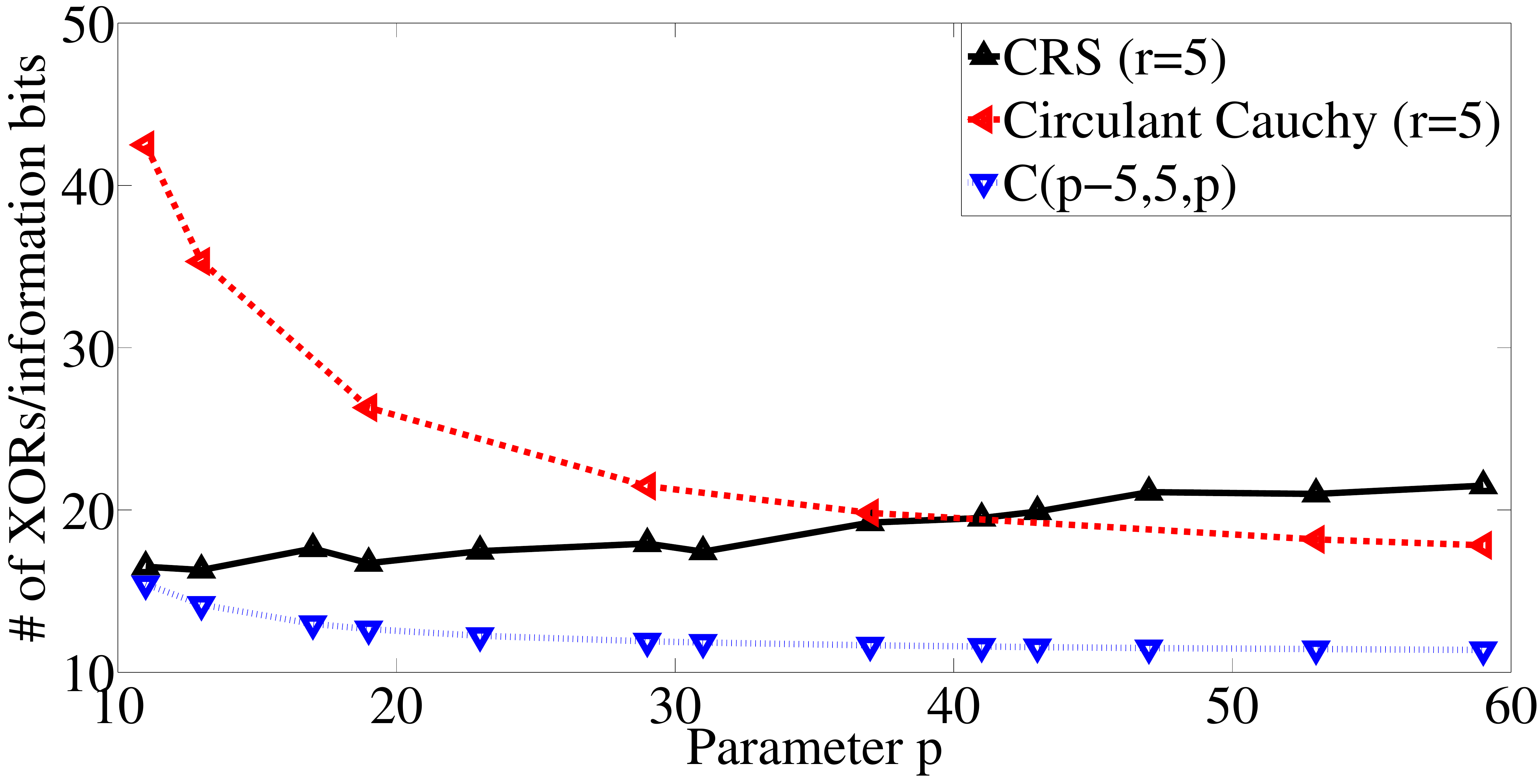}
\par\end{centering}

}\caption{The normalized decoding complexity. The complexity has been normalized to the
number of information bits.}
\vspace{-0.4cm}
\label{de:complexity}
\end{figure}

The normalized decoding complexities of $r=4$ and $r=5$ are shown in Fig. \ref{de:complexity}. We observe that
the decoding complexity of CRS codes increases as $p$ increases, and the decoding complexity of Circulant
Cauchy codes decreases while $p$ increases, where $r$ is fixed. However, the normalized decoding complexity
of $\mathcal{C}(p-r,r,p)$ is almost the same for different values of $p$ when $r$ is constant.
In general, the decoding complexity of $\mathcal{C}(p-r,r,p)$ is much less than that of CRS codes and Circulant
Cauchy codes, and the complexity difference between $\mathcal{C}(p-r,r,p)$ and CRS codes becomes larger when $p$ increases.
When $r=4$, the percentage of improvement over CRS codes and Circulant Cauchy codes varies between 15.4\% and 47.9\%, and
33.6\%-60.5\%, respectively. When $r=5$, the percentage of improvement over CRS codes and Circulant Cauchy codes
varies between 6.5\% and 47.1\%, and 36.2\%-63.7\%, respectively.

\section{Conclusion}

\label{sec:discussions}
We propose a construction of Cauchy array codes over a specific binary cyclic ring which employ XOR and bit-wise cyclic shifts. These codes have been proved with MDS property. We present an $\mathbf{LU}$ factorization of Cauchy matrix over the binary cyclic ring and propose an efficient decoding algorithm based on the $\mathbf{LU}$ factorization of Cauchy matrix.
We show that the proposed Cauchy array code improve the encoding complexity and decoding complexity over existing codes.

We conclude with few future work. In the constructed array codes, the parameter $p$ is restricted to be a special class of prime number. It could be interesting to find out whether there exist MDS Cauchy array codes without this restriction. When there is a single column fails, the total number of bits downloaded from the surviving columns is termed as repair bandwidth. How to recover the failed column with repair bandwidth as little as possible is another interesting future work.

\appendices

\bibliographystyle{IEEEtran}
\bibliography{IEEEabrv,CNC}

\end{document}